\documentclass[a4paper,USenglish,cleveref,autoref,thm-restate]{lipics-v2021}
\pdfoutput=1
\usepackage{mathtools}
\usepackage{amsthm}

\newtheorem{problem}[theorem]{Problem}

\title{Nonnegativity Problems for Matrix Semigroups} 

\titlerunning{Nonnegativity Problems for Matrix Semigroups} 

\author{Julian D'Costa}{Department of Computer Science, University of Oxford, UK}%
{julianrdcosta@gmail.com}{https://orcid.org/0000-0003-2610-5241}{emmy.network foundation under the aegis of the Fondation de Luxembourg.}

\author{Jo\"el Ouaknine}%
{Max Planck Institute for Software Systems, Saarland Informatics Campus, Germany}%
{joel@mpi-sws.org}{https://orcid.org/0000-0003-0031-9356}{%
DFG grant 389792660 as part of TRR 248 (see
\url{https://perspicuous-computing.science}). 
Jo\"el Ouaknine is also affiliated with Keble College, Oxford as \href{http://emmy.network/}{\texttt{emmy.network}} Fellow.}

\author{James Worrell}%
{Department of Computer Science, University of Oxford, UK}%
{jbw@cs.ox.ac.uk}{https://orcid.org/0000-0001-8151-2443}{}
\authorrunning{J. D'Costa, J. Ouaknine and J. Worrell} 
\ccsdesc[500]{Computing methodologies~Symbolic and algebraic manipulation}
\ccsdesc[500]{Theory of computation~Formal languages and automata theory}
\keywords{Decidability, Linear Recurrence Sequences, Schanuel's Conjecture} 

\nolinenumbers 

\hideLIPIcs  

\EventEditors{John Q. Open and Joan R. Access}
\EventNoEds{2}
\EventLongTitle{42nd Conference on Very Important Topics (CVIT 2016)}
\EventShortTitle{CVIT 2016}
\EventAcronym{CVIT}
\EventYear{2016}
\EventDate{December 24--27, 2016}
\EventLocation{Little Whinging, United Kingdom}
\EventLogo{}
\SeriesVolume{42}
\ArticleNo{23}

\def\Z{\mathbb{Z}}
\def\R{\mathbb{R}}
\def\Q{\mathbb{Q}}
\def\C{\mathbb{C}}

\newcommand{\PTIME}{\textbf{PTIME}}

\newcommand{\abs}[1]{\left\lvert #1\right\rvert}

\def\set#1{{\{ #1 \}}}

\providecommand{\st}{}
\renewcommand{\st}{\mathrel{\mid}}

\renewcommand{\int}{\mathbb{Z}}

\DeclareMathOperator{\diag}{diag}
\DeclareMathOperator{\poly}{poly}
\DeclareMathOperator{\MPoly}{MPoly}
\begin{document}
	
	\maketitle
	
	\begin{abstract}
		The matrix semigroup membership problem asks, given square matrices $M,M_1,\ldots,M_k$
of the same dimension, whether $M$ lies in the semigroup generated by $M_1,\ldots,M_k$.  It is classical that this problem is undecidable in general but decidable in case $M_1,\ldots,M_k$ commute.  In this paper we consider the problem of whether, given $M_1,\ldots,M_k$, the semigroup
generated by $M_1,\ldots,M_k$ contains a non-negative matrix.  We show that in case $M_1,\ldots,M_k$ commute, this problem is decidable subject to Schanuel's Conjecture.  We show also that the problem is undecidable if the commutativity assumption is dropped. 
A key lemma in our decidability result is a procedure to determine, given a matrix $M$, whether the sequence of matrices $(M^n)_{n\geq 0}$ is ultimately nonnegative.  This answers a problem posed by S. Akshay~\cite{DBLP:conf/cade/AkshayCP22}.
The latter result is in stark contrast to the notorious fact that it is not known how to determine effectively whether for any specific matrix index $(i,j)$ the sequence $(M^n)_{i,j}$ is ultimately nonnegative (which is a formulation of the Ultimate Positivity Problem for linear recurrence sequences).
	\end{abstract}
\newpage

\section{Introduction}
The \emph{Membership Problem} for finitely generated matrix semigroups 
asks, given square matrices $M_,M_1,\ldots,M_k$ of the same dimension and with rational entries, whether $M$ lies in the semigroup generated by $M_1,\ldots,M_k$.  The problem was shown to be undecidable by Markov in the 1940s~\cite{Markov}, becoming thereby one of the first instances of a natural undecidable mathematical problem.  The problem however becomes decidable under the assumption that the matrices $M_1,\ldots,M_k$ commute~\cite{Babai}.

There are many variants of the Membership Problem.  In the \emph{Mortality Problem} one asks whether the zero matrix lies in a finitely generated matrix semigroup.  This problem is undecidable already for $3 \times 3$ matrices~\cite{Paterson}.  Meanwhile the \emph{Identity Problem} asks to determine whether the identity matrix lies in a given finitely generated matrix semigroup.  The latter problem is undecidable in general but decidable for certain nilpotent and low-order matrix groups~\cite{BHP17,Dong22,Dong23,KNP18}.

This paper is concerned with the 
\emph{Non-negative Membership Problem}, which asks to determine whether a given finitely generated matrix semigroup contains a non-negative matrix, i.e., a matrix all of whose entries are non-negative.  We show that this problem is undecidable in general but is decidable in the commutative case subject to Schanuel's Conjecture, a well-known unifying conjecture in transcendence theory.  Our reliance on Schanuel's Conjecture arises because we reduce the 
commutative case of the 
Non-negative Membership Problem to a decision problem in the first-order theory of the reals with exponential.  As shown by Macintyre and Wilkie~\cite{WM}, the latter theory is decidable subject to Schanuel's Conjecture.

A key lemma in our main decidability result involves determining for a given matrix $M$ whether for all but finitely many $n \in \mathbb{N}$ the matrix power $M^n$ is non-negative.  In such a case we say that $M$ is eventually non-negative.
We give an effective characterisation of eventually non-negative matrices, answering an open question of S.~Akshay~\cite{DBLP:conf/cade/AkshayCP22}.  The characterisation is relatively straightforward and relies on classical results about rational sequences over the semi-ring of non-negative rational numbers.
We note that the problem of determining 
whether for some fixed index $(i,j)$ 
the sequence of scalars
$(M^n)_{i,j}$ is ultimately non-negative is a formulation of the Ultimate Positivity Problem for linear recurrence sequences, decidability of which is a longstanding open problem~\cite{Ouaknine2014}.

Note that a finitely generated matrix semigroup contains a non-negative matrix if and only if it contains an eventually non-negative matrix.  Using a symbolic version of our criterion for determining whether a given matrix is eventually non-negative, we reduce the Non-negative Membership Problem to the 
decision problem for the theory of reals with exponential.

A simpler variant of our main result concerns the problem 
of deciding whether a finitely generated matrix semigroup contains a positive matrix.  We likewise show that this problem is decidable in the case of commuting matrices, subject to Schanuel's Conjecture.
Here we rely on a known characterisation of eventually positive matrices, due to~\cite{noutsos2006perron}.

As far as we are aware the Non-negative Membership Problem 
has not been directly addressed before.  We note however that 
decidability of the problem for sub-semigroups of 
the group $\mathbf{GL}(2,\mathbb{Z})$ of $2\times 2$ invertible integer matrices follows directly from~\cite[Theorem 13]{colcombet}.

\section{Mathematical Background}

Here we state some number-theoretic results that we will need in the sequel. 

First stated in the 1960s, Schanuel's conjecture is a unifying conjecture in transcendental number theory that generalizes many of the classical results in the field.

\begin{conjecture}[Schanuel's conjecture \cite{Lang1966}]
	 If $\alpha_1, \dots, \alpha_k \in \C$ are rationally linearly independent, then some $k$-element subset of $\set{\alpha_1, \dots, \alpha_k, e^{\alpha_1}, \dots, e^{\alpha_k}}$ is algebraically independent.
\end{conjecture}

An elementary point is an element of $\mathbb{C}^n$ that is an isolated, nonsingular solution of $n$ equations in $n$ variables, with each equation being either of the form $p = 0$, where $p$ is a polynomial in $\Q[x_1,\dots,x_n]$, or of the form $x_j - e^{x_i} = 0$. An elementary number is the polynomial image of an elementary point. 

Intuitively, an elementary number is obtained by starting with the rationals and implicitly or explicitly applying addition, subtraction, multiplication, division, exponentiation, and taking natural logarithms. In constrast to the case with algebraic numbers, deciding whether an elementary number is equal to zero is not straightforward. 

\begin{proposition}[Richardson \cite{RICHARDSON1997627}]
\label{richardson}
	The problem of determining zeroness of an elementary number is semi-decidable.  The problem is moreover decidable if one assumes 
 Schanuel's conjecture.
\end{proposition}

We will also need the following theorem due to Masser.
\begin{theorem}[Multiplicative relations among algebraic numbers \cite{Masser1988})]
	Let $m$ be fixed, and let $\lambda_1, \dots, \lambda_m$ be complex algebraic numbers. Consider the free abelian group $L$ under addition given by
	$$L = \set{(v_1, \dots, v_m) \in \Z^m : \lambda_1^{v_1} \dots \lambda_m^{v_m} = 1}.$$
	$L$ has a basis $\set{\vec{\ell_1}, \dots, \vec{\ell_p}} \subseteq \Z^m$ (with $p \leq m$), where the entries of each of the $\vec{\ell_j}$ are all polynomially bounded in 
 the sum of the heights and degrees of the minimal polynomials of  $\lambda_1,\ldots,\lambda_m$.
\end{theorem}

\section{Linear Recurrence Sequences}
First, we recall some basic terminology and results about
linear recurrence sequences.

A sequence $\boldsymbol{u} = (u_n)_{n\geq 0}$ of elements of a
semiring $K$ is called \emph{$K$-rational} if there exists $d\geq 1$,
$v,w \in K^d$ and $M\in K^{d\times d}$ such that $u_n=v^\top M^n w$
for all $n$.  When $K$ is a field, a sequence is $K$-rational if and
only if it satisfies a linear recurrence relation \[ u_n = a_1u_{n-1}
+\cdots a_du_{n-d} \quad (n\geq d) \] where $a_1,\ldots,a_d \in K$.
In this case we also call $\boldsymbol{u}$ a \emph{linear recurrence
sequence (LRS).}   

With the unique minimal order recurrence satisfied by $\boldsymbol{u}$
we associate the \emph{characteristic polynomial}
\[ P(X) = X^d - a_1X^{d-1} - \cdots - a_d \, . \] The roots of $P(X)$
are called the \emph{characteristic roots} of $\boldsymbol{u}$.
Writing $\lambda_1,\ldots,\lambda_m$ for the distinct characteristic
roots, in non-increasing order of absolute value, the sequence
$\boldsymbol{u}$ admits a closed-form representation:
\[ u_n = \sum_{i=1}^m P_i(n) \lambda_i^n \, ,\] where the $P_i$ are
univariate polynomials whose coefficients are algebraic over $K$.  We
say that $\boldsymbol{u}$ is \emph{non-degenerate} if no quotient of
two distinct characteristic roots is a root of unity.  We also say
that a matrix $M \in \Q^{d\times d}$ is non-degenerate if no quotient
of two distinct eigenvalues is a root of unity.

In this paper we exclusively consider sequences with rational entries.
We say that an LRS $\boldsymbol{u}$ is \emph{dominated} if $\lambda_1$
is the unique characteristic root of maximum modulus.  Note that in
this case $\lambda_1$ is necessarily real.  We have the following three
straightforward propositions concerning dominated LRS.

\begin{proposition}
If an LRS $\boldsymbol{u}$ is dominated then $\{ n \in
\mathbb{N}: u_n \geq 0\}$ is an effectively computable ultimately
periodic set.
\label{prop:DOM1}
\end{proposition}
\begin{proof}
Consider the closed form representation
$u_n = \sum_{i=1}^m P_i(n)\lambda_i^n$ with unique
dominant root $\lambda_1$, necessarily real.  Suppose that $P_1$ has degree $k$ and
leading coefficient $a$.  Then we have
$\frac{u_n}{n^k \left|\lambda_1\right|^n} = a(\lambda_1/\left|\lambda_1\right|)^n +
o(1)$.  Hence for sufficiently large $n$ the sign of $u_n$ is
a function of the parity of $n$.
  \end{proof}

\begin{proposition}
If $\boldsymbol{u}$ is a non-degenerate LRS such that some subsequence
$(u_{cn+d})_{n \geq 0}$ is dominated, where $c$ is a positive integer
and $d\in \{0,1\ldots,c-1\}$, then $\boldsymbol{u}$ itself is
dominated.
\label{prop:DOM2}
\end{proposition}
\begin{proof}
The sequence $\boldsymbol{u}$ admits a closed-form representation 
$u_n = \sum_{i=1}^m P_i(n)\lambda_i^n$, where 
$\lambda_1,\ldots,\lambda_m$ are the characteristic roots and 
$P_1,\ldots,P_m$ 
are polynomials.  Then
\begin{eqnarray*}
u_{cn+d} &=& \sum_{i=1}^m P_i(cn+d)\lambda_i^{cn+d}\\
               &=& \sum_{i=1}^m Q_i(n) (\lambda_i^c)^n \, 
	       \end{eqnarray*}
	       where $Q_i(n):=P_i(cn+d)\lambda_i^d$ for $i\in 
	       \{1,\ldots,m\}$.

Note that the polynomials $Q_1,\ldots,Q_m$ are non-zero and, by
non-degeneracy of $\boldsymbol{u}$, the numbers
$\lambda_1^c,\ldots,\lambda_m^c$ are pairwise distinct.  Since the
sequence $(u_{cn+d})_{n\geq 0}$ is dominated, we have that
$\lambda_1^c$ is its unique characteristic root of maximum modulus.
But then $\lambda_1$ is the unique characteristic root of $\boldsymbol{u}$
\end{proof}

\begin{proposition}
  An LRS that is both non-degenerate and rational over the semiring
  $\mathbb{Q}_+$ of nonnegative rational numbers is
  dominated.
  \label{prop:DOM3}
\end{proposition}
\begin{proof}
Berstel~\cite{Berstel1971} showed that if a sequence $\boldsymbol{u}$ is 
$\mathbb{Q}_+$-rational then its characteristic roots of maximum
modulus all have the form $\rho \omega$ for some non-negative real
number $\rho$ and root of unity $\omega$.  Since $\boldsymbol{u}$ is
non-degenerate it follows that it has a unique dominant root. For an exposition, see \cite[Chap. 8, Thm 1.1]{berstel2011noncommutative}. We provide an alternate proof based on the Perron-Frobenius theorem in Appendix \ref{sec:positive-dom}.
\end{proof}

\begin{theorem}
Given $M \in \Q^{d\times d}$ the set $S:=\{ n \in \mathbb{N} :
M^n \geq 0\}$ is ultimately periodic and effectively computable.
\label{thm:powers}
\end{theorem}
\begin{proof}
  Recall that for some (effectively computable) strictly positive
  integer $L$ the matrix $M^L$ is non-degenerate.  It will suffice to
  show that for each $l \in \{0,\ldots,L-1\}$ we can compute the set
  $S_l := \{ n \in S : n\equiv l \bmod L\}$.  Our procedure to do this
  is as follows.  First, check for every pair of indices 
  $i,j\in\{1,\ldots,d\}$, whether the sequence $(u^{(i,j)}_n)_{n\geq
    0}$ given by $u^{(i,j)}_n = (M^{Ln+l})_{i,j}$,  is dominated.  If
  yes then by Proposition~\ref{prop:DOM1} we
  can compute $S_l$ as the intersection over all pairs $(i,j)$ of the sets $\{ n\in\mathbb{N} :
  u^{(i,j)}_n \geq 0\}$.
  If no, then we claim that $S_l$ is empty.

  Indeed, suppose
  $n_0 \in S_l$.   Then for each pair of indices 
  $i,j\in\{1,\ldots,d\}$, the LRS $(v^{(i,j)}_n)_{n\geq 0}$ defined by
  $v^{(i,j)}_n = (M^{n_0(1+Ln)})_{i,j}=e_i^\top M^{n_0} (M^{Ln_0})^n e_j^\top$ is both non-degenerate and
  $\mathbb{Q}^+$-rational.  By Proposition~\ref{prop:DOM3} each sequence
  $(v^{(i,j)}_n)_{n\geq 0}$ is dominated.  Moreover, since
  $(v^{(i,j)}_n)_{n\geq 0}$ is a subsequence of
  $(u^{(i,j)}_n)_{n\geq 0}$, the latter is also dominated by Proposition~\ref{prop:DOM2}.
This proves (the contrapositive of) our claim.
  \end{proof}

  \begin{remark}
    We can extract from the proof of Theorem~\ref{thm:powers} an
    effective characterisation of those matrices $M$ such that $M^n
    \geq 0$ for some positive integer $n$.
Let $L$ be the least
    positive integer such that $M^L$ is non-degenerate.  Then some
    positive power of $M$ is a non-negative matrix iff some positive power of $M^L$ is
    non-negative iff for all indices $(i,j)$ the sequence
    $(u^{(i,j)}_n)_{n\geq 0}$ defined by
    $u^{(i,j)}_n:= (M^{Ln})_{i,j}$ is dominated and not ultimately negative.
\label{rem:non-negative-power}
  \end{remark}

\section{The Positive Membership Problem}

\begin{problem}[Positive Membership for Commutative Semigroup]
	Given a set of commuting $d \times d$ matrices $\set{A_1, \dots, A_k}$ with rational entries, decide whether the semigroup generated by multiplying these matrices together contains a matrix with all its entries strictly greater than zero.
\end{problem}

\begin{definition}[Eventually Positive Matrix]
	We call a matrix $M$ \emph{eventually positive} if there exists a natural number $n_0$ such that for all $n \geq n_0$, every entry of the matrix $M^n$ is greater than zero.
\end{definition}

We will need the following result adapted from Noutsos \cite{noutsos2006perron} which characterizes eventual positivity by proving a converse to the Perron-Frobenius theorem.

\begin{definition}[strong Perron-Frobenius property]
	A matrix $A \in \R^{n \times n}$ has the \emph{strong Perron-Frobenius property} if there exists an eigenvalue $\lambda_1$ with the following properties:
	\begin{enumerate}
		\item $\lambda_1$ is strictly dominant
		\item $\lambda_1$ is real and positive
		\item $\lambda_1$ is a simple eigenvalue
		\item the corresponding eigenvector $v^{(1)}$ can be chosen to have all positive entries.
	\end{enumerate}
\end{definition}

Note that Noutsos does not explicitly include simple but this is necessary because otherwise the identity matrix is a counterexample to the theorem characterizing eventual positivity.

\begin{theorem}[Converse to Perron-Frobenius \cite{noutsos2006perron}]
	Let $A$ be a matrix satisfying the strong Perron-Frobenius property with dominant eigenvalue $\lambda_1$ and positive eigenvector $p$ such that $A^\top$ also satisfies the strong Perron-Frobenius property with positive eigenvector $q$. Then $\lim_{n \rightarrow \infty} (A/\lambda_1)^n = \frac{pq^\top}{q^\top p} > 0$. 
\end{theorem}

This can be used to show:
\begin{theorem}[Characterizing eventual positivity \cite{noutsos2006perron}]
	A matrix $A$ is eventually positive iff $A$ and $A^\top$ both have the strong Perron-Frobenius property.
\end{theorem}

Thus it is enough for us to search for a matrix in the semigroup that (along with its transpose) has the strong Perron-Frobenius property. We simply power it to obtain a positive matrix. Conversely, because a positive matrix is trivially eventually positive, the absence of such a matrix implies that the semigroup does not contain a positive matrix.

\subsection{Understanding product eigenspaces}

We now consider a decomposition of $\C^d$ induced by the commuting matrices $A_1, \dots, A_k \in \C^{d \times d}$. Note that commuting matrices have the same eigenvectors. Let $\sigma(A_i)$ denote the set of eigenvalues of $A_i$. We will now construct tuples of eigenvalues of the $A_i$ that we can use to decompose $\C^n$ into eigenspaces.

Let $$\boldsymbol{\lambda} = (\lambda_1, \dots, \lambda_k) \in \sigma(A_1) \times \dots \times \sigma(A_k)$$ be a tuple of eigenvalues. Recall that $\ker(A_i - \lambda_i I)^d$ is the generalized eigenspace of $\lambda_i$ of $A_i$. We say that the tuple $\boldsymbol{\lambda}$ is \emph{compatible} if  

$$V_{\boldsymbol{\lambda}} := \bigcap_{i=1}^k \ker(A_i - \lambda_i I)^d$$ 
is not $\set{\boldsymbol{0}}$. Intuitively, compatible tuples are formed of eigenvalues of the $A_i$ that correspond to overlapping eigenspaces. Call the set of compatible tuples $\Sigma$. It can be shown \cite{joao2016solvability} that $$\C^n = \oplus_{\boldsymbol{\lambda} \in \Sigma} V_{\boldsymbol{\lambda}}.$$

From the Perron-Frobenius theorem \cite{Meyer2023}, we know that it is only possible to have a positive matrix in the semigroup if there is exactly one subspace $V_{\boldsymbol{\lambda}}$ which is 1-dimensional and has a basis vector with all positive entries. Call this subspace the \emph{positive eigenspace}.

Assume we have such an eigenspace for both $\set{A_1, \dots, A_k
}$ and their transposes. This satisfies conditions 3 and 4 of the strong Perron-Frobenius property. We now need conditions 1 and 2: to check whether it is possible to form a product in such a way that the positive eigenspace has a larger eigenvalue than all other eigenspaces.

Consider such a product $A_1^{m_1} \dots A_k^{m_k}$. 
Let the notation $\lambda_{rl}$ denote the $l$th eigenvalue (with multiplicity) of $A_r$, where $r$ ranges over the matrices $A_1$ to $A_k$ and $l$ ranges over the eigenspaces. Let $p$ be the index of the positive eigenspace. We want to check whether there exists a tuple $\boldsymbol{m} = (m_1, \dots, m_k)$ such that the eigenvalue of the positive eigenspace $\prod_{r=1}^k\lambda_{rp}^{m_r}$ is positive real and strictly dominant, ie:

\begin{enumerate}
	\item $\prod_{r=1}^k\lambda_{rp}^{m_r} > 0 ,$\\
	\item $\forall l \neq p, \; \abs{\prod_{r=1}^k\lambda_{rp}^{m_r}} > \abs{\prod_{r=1}^k\lambda_{rl}^{m_r}}$.
\end{enumerate}

Note that these are equivalent to the conditions $\Im(\sum_{r=1}^k m_r \log (\lambda_{rp})) = 0\!\! \mod 2\pi$ and $\forall l \neq p, \; \sum_{r=1}^k m_r \Re(\log (\lambda_{rl}/\lambda_{rp})) < 0 $. We use the principal branch of the complex logarithm.

Let $\boldsymbol{c}$ be the $k$-dimensional vector defined by $c_r = \Im(\log(\lambda_{rp}))$. Let $\boldsymbol{A}$ be the $(d -1 \times k)$-matrix defined by $a_{lr} = \Re(\log(\lambda_{rl}/\lambda_{rp})$. Here $l$ runs over elements of $\set{1,\dots,d}$ apart from $p$). The two conditions are equivalent to a positive integer solution $\boldsymbol{m}$ to the following linear program: $$(\boldsymbol{c}^\top \boldsymbol{m} = \boldsymbol{0}\!\!\!\! \mod 2\pi) \wedge \boldsymbol{A}\boldsymbol{m} < \boldsymbol{0}. $$

Note that by incorporating positivity constraints $-m_i < 0$ in $\boldsymbol{A}$ we may search for an integer solution. We now show that such a program is decidable, subject to Schanuel's conjecture.
\subsection{Integer Programming with Logs of Algebraic Numbers}

\begin{problem}[strict homogenous IP-log]
	Given a matrix $A \in \R^{m \times n}$ and vector $c \in \R^n$ with entries that are logs of algebraic numbers, does there exist a vector $x \in \Z^n$ such that $c^\top x = 0 \mod 2\pi$ and $Ax < 0$ (entry-wise) ?
\end{problem}

Using Masser's theorem on multiplicative relations among algebraic numbers  we can find an integer basis $\set{\vec{v_1},\dots,\vec{v_{l}}}$ (where $l \leq n$) of solutions to $c^\top x = 0 \mod 2\pi$ (in polynomial time). Let $B \in \Z^{n \times l}$ be the matrix with these basis vectors as columns. Then $c^\top x = 0 \mod 2\pi$ iff $x = By$ for some integer vector $y$. Our problem is equivalent to asking whether there exists a vector $y \in \Z^{l}$ such that $AB y < 0$ (entry-wise). Thus without loss of generality we can assume that $c$ is the zero vector. 

Now observe that the set of real solutions $\set{x \in \R^n \st Ax < 0}$, which we will now denote by $C$, is an \emph{open convex cone}. This strong geometric property has the following convenient consequence: 

\begin{proposition}[Cone Real2Int]
	If $C := \set{x \in \R^n \st Ax < 0}$ is non-empty, it must contain an integer point.
\end{proposition}
\begin{proof}
	That $C$ is an open convex cone follows easily from the definition and linearity of $A$. We will use the fact that $C + C = C$ (where $+$ denotes Minkowski set addition).
	
	Assume $C$ is non-empty. Let $p$ be a point in $C$. By openness, there exists $\varepsilon > 0$ such that the full-dimensional ball $B_\varepsilon(p) \subseteq C$. By adding together $O(\lceil 1/ \varepsilon \rceil)$ copies of $B_\varepsilon(p)$ we see that $C$ contains a ball large enough that it must contain an integer point.
\end{proof}

\begin{theorem}
	The strict homogenous IP-log problem is decidable, assuming Schanuel's conjecture is true.
\end{theorem}

\begin{proof}
	An instance of the strict homogenous IP-log problem asks to determine the truth of the sentence: 
	$$ \exists x \in \Z^n : (c^\top x = 0\!\!\! \mod 2\pi) \wedge Ax < 0$$
	By Masser's theorem we can eliminate the first conjunct. By the previous proposition we can reduce to solving the linear program $Ax < 0$ over the reals. The only difficulty is that the entries of $A$ are logs of algebraic numbers. We can still apply Fourier-Motzkin elimination.
	Fourier-Motzkin elimination is a simple method for eliminating variables from systems of linear inequalities \cite{dantzig1963linear}. The procedure consists of isolating one variable at a time by dividing all inequalities by its coefficient and checking whether all the lower bounds thus derived for it are less than the upper bounds. Note that this method preserves the set of solutions on the remaining variables, so a solution of the reduced system can always be extended to a solution of the original one. 
	At each step, generating the new set of constraints requires correctly determining the sign of each coefficient generated in the previous step in order to know whether the inequality defines a lower bound or an upper bound. These coefficients are rational functions in logs of algebraic numbers. As such, they are elementary numbers, so using Richardson's algorithm (Prop \ref{richardson}), we can determine which coefficients are zero. For a coefficient which is not zero, we need merely compute it to sufficient precision to separate our approximation from zero.
\end{proof}

The following example illustrates the main points of the argument above:
Let $\omega = e^{2\pi i /3}$ be a primitive cube root of unity and let $\lambda_1, \lambda_2, \lambda_3, \lambda_4$ be real positive algebraic numbers. Consider the following system of inequalities:
\begin{align*}
	\Im(x_1 \log(\omega) + x_2 \log(\omega^2)) &= 0 \mod 2\pi  \\
	\wedge \;\Re(x_1 \log(\lambda_1) + x_2 \log(\lambda_2)) &< 0  \\
	\wedge \; \Re(x_1 \log(\lambda_3) + x_2 \log(\lambda_4)) &< 0
\end{align*}
Using Masser's theorem (or in this instance some algebraic manipulation) we see that the first conjunct is satisfied when $x_1 = 3y_1 + 2y_2$ and $x_2 = 3y_1 -y_2$ for some $y_1, y_2 \in \Z$. Thus we can reduce to the following system:
\begin{align*}
	 & (3y_1 + 2y_2)\log\lambda_1 + (3y_1 -y_2) \log\lambda_2 < 0 \wedge  (3y_1 + 2y_2) \log\lambda_3 + (3y_1 -y_2)  \log\lambda_4 < 0\\
	 &\iff y_1 \log(\lambda_1^3 \lambda_2^3) + y_2 \log(\lambda_1^2/ \lambda_2) < 0 \wedge  y_1 \log(\lambda_3^3 \lambda_4^3) + y_2 \log(\lambda_3^2/ \lambda_4) < 0
\end{align*}

Assuming $\lambda_1\lambda_2 > 1 > \lambda_3\lambda_4$ we can eliminate $y_1$ by dividing out the coefficients (with known signs):
$$y_1  < - y_2 \frac{\log(\lambda_1^2/ \lambda_2)}{\log(\lambda_1^3 \lambda_2^3)} \wedge  y_1  > - y_2 \log\frac{(\lambda_3^2/ \lambda_4)}{\log(\lambda_3^3 \lambda_4^3)}.$$

The system has a solution if $$-y_2\log\frac{(\lambda_3^2/ \lambda_4)}{\log(\lambda_3^3 \lambda_4^3)} < - y_2 \frac{\log(\lambda_1^2/ \lambda_2)}{\log(\lambda_1^3 \lambda_2^3)},$$ ie if $$ \frac{\log(\lambda_1^2/ \lambda_2)}{\log(\lambda_1^3 \lambda_2^3)} - \frac{\log(\lambda_3^2/ \lambda_4)}{\log(\lambda_3^3 \lambda_4^3)} \neq 0,$$ which can be determined by Richardson's algorithm. (In this case existence of a solution was equivalent to the matrix $A$ having non-zero determinant, but this will not hold for general systems with more constraints.)

\subsection{Algorithm}

We thus have the following algorithm for the positive membership problem in the commutative case:

Given a set of commuting rational matrices $\set{A_1, \dots, A_k}$:

\begin{enumerate}
	\item Simultaneously upper-triangularize the matrices (see \cite{CommutingMatrices} for a method) and check for a one-dimensional positive eigenspace. If none or more than one exists, return NO. Do the same for the transposes.
	\item Having identified the positive eigenspace $p$, compute the corresponding $A$ and $c$ constraints for the IP-log problem as described above. Add positivity constraints $x_i > 0$ to $A$. If the IP-log problem is unsatisfiable, return NO. Otherwise, return YES.
	\item In order to handle cases where some matrix has a zero power, repeat the algorithm for subsets of the original $k$ matrices.
\end{enumerate}

We conclude:

\begin{theorem}
	The positive membership problem is decidable for commutative semigroups, assuming Schanuel's conjecture is true.
\end{theorem}

\section{The Non-negative Membership Problem for Commutative Semigroups}

We now combine the ideas in the previous two sections to solve the problem of deciding whether a semigroup of commuting matrices contains a non-negative matrix. For ease of exposition we will assume that the matrices are simultaneously diagonalizable. The general commuting case involves more complicated algebra and is proven in Appendix \ref{sec:nondiag}.

\begin{problem}[Non-negative Matrix in Commutative Simultaneously Diagonalizable Semigroup]
	Given a set of commuting simultaneously diagonalizable $d \times d$ matrices $\set{A_1, \dots, A_k}$ with rational entries, decide whether the semigroup generated by multiplying these matrices together contains a matrix with all its entries greater than or equal to zero.
\end{problem}

First, we refine our notion of positively dominated recurrences.
\begin{definition}[positively dominated by $p$]
	Let $u_n$ be a linear recurrence. For simplicity we assume $u_n$ is non-degenerate and does not have any polynomial terms in its exponential-polynomial form, which is then $\sum_{l=1}^d c_l\lambda_l^n$ for complex numbers $\lambda_1,\dots, \lambda_d$ and coefficients $c_l$. 
	We say $u_n$ is \emph{positively dominated by term $p$} (here p (for positive) refers to the index) if 
	\begin{enumerate}
		\item $c_p > 0$\\
		\item $\lambda_p > 0$\\
		\item $\forall l \neq p, \; \abs{\lambda_p} > \abs{\lambda_l} $. 
	\end{enumerate}
	Call this predicate $PD_p(u_n)$.
\end{definition}

Given a single matrix $M$, we define the $d^2$ recurrences $u^{ij}_{n} := e_i^\top M^n e_j = \sum_{l=1}^r c_l\lambda_l^n$ where $\lambda_l$ are the eigenvalues of $M$.

We have shown in previous sections that $M$ being eventually non-negative is equivalent to the decidable condition $$\forall i \forall j (u^{ij}_n \text{ is ultimately zero} \vee \exists p \; PD_p(u^{ij}_n) ).$$

We now show a similar construction for multiple matrices.
Let $A_1,\dots,A_k \in \Q^{d \times d}$ be a set of commuting simultaneously diagonalizable matrices. The idea is to search for an eventually non-negative matrix $A_1^{m_1} \dots A_k^{m_k}$.

Define the $(m_1,\dots,m_k)$-integer parameterized matrix entry recurrence $u_n^{ij}(m_1,\dots,m_k) := e_i^\top [A_1^{m_1} \dots A_k^{m_k}]^n e_j$.

Let $m$ denote the tuple $(m_1,\dots,m_k)$. The existence of an eventually non-negative matrix (and thus, a non-negative matrix) in the semigroup is equivalent to the decidable condition $$ENN := \exists m \; \forall i \forall j (u^{ij}_n(m) \text{ is ultimately zero} \vee \exists p \; PD_p(u^{ij}_n(m)) ).$$

Let $S$ be a matrix that simultaneously diagonalizes the matrices $A_1,\dots,A_k$ such that $A_r = S^{-1} D_r S = S^{-1} \diag(\lambda_{r1},\dots,\lambda_{rd}) S$. The notation $\lambda_{rl}$ denotes the $l$th eigenvalue of $A_r$ (with multiplicity).

Then
\begin{align*}
	u_n^{ij}(m_1,\dots,m_k) &:= e_i^\top [A_1^{m_1} \dots A_k^{m_k}]^n e_j \\
	&= e_i^\top S^{-1} [D_1^{m_1} \dots D_k^{m_k}]^n S e_j \\
	&= e_i^\top S^{-1} [\diag(\lambda_{11},\dots,\lambda_{1d})^{m_1} \dots \diag(\lambda_{k1},\dots,\lambda_{kd})^{m_k}]^n S e_j \\
	&= e_i^\top S^{-1} \left[\diag \left( \prod_{r=1}^k\lambda_{r1}^{m_r}, \dots, \prod_{r=1}^k\lambda_{rd}^{m_r} \right )\right]^n S e_j \\
	&= \sum_{l=1}^d  (S^{-1})_{il}\, (S)_{lj} \left[\prod_{r=1}^k\lambda_{rl}^{m_r}\right]^n .
\end{align*}

Now $u_n^{ij}(m_1,\dots,m_k)$ is an exponential-polynomial recurrence in the complex numbers $\prod_{r=1}^k\lambda_{r1}^{m_r},\dots,\prod_{r=1}^k\lambda_{rd}^{m_r}$ with coefficients $(S^{-1})_{il} (S)_{lj}$.

We see that $u_n^{ij}(m_1,\dots,m_k)$ is positively dominated by $p$ if
\begin{enumerate}
	\item $(S^{-1})_{ip}\, (S)_{pj} > 0,$\\
	\item $\prod_{r=1}^k\lambda_{rp}^{m_r} > 0,$\\
	\item $\forall l \neq p \text{ such that } (S^{-1})_{ip}\, (S)_{pj} \neq 0,  \; \abs{\prod_{r=1}^k\lambda_{rp}^{m_r}} > \abs{\prod_{r=1}^k\lambda_{rl}^{m_r}}$.
\end{enumerate}

As in the previous section, let $\boldsymbol{c}(p)$ be the $k$-dimensional vector defined by $c_r = \Im(\log(\lambda_{rp}))$. Let $\boldsymbol{A}(p)$ be the (at most) $d -1 \times k$-matrix defined by $a_{lr} = \Re(\log(\lambda_{rl}/\lambda_{rp})$. Here $l$ runs over elements of $\set{1,\dots,d}$ apart from $p$ such that $(S^{-1})_{il}\, (S)_{lj} \neq 0$.
Let $\boldsymbol{m}$ be the vector $(m_1,\dots,m_k)$.

Now condition 2 is equivalent to $\boldsymbol{c}(p)^\top \boldsymbol{m} = 0\!\! \mod 2\pi$ and condition 3 is equivalent to $\boldsymbol{A}(p)\boldsymbol{m} < 0$.

For simplicity we introduce the following new notation:
\begin{align*}
    Z(i,j)  &:= u^{ij}_n \text{ is identically zero}\\
    S(i,j,p) &:= (S^{-1})_{ip}\, (S)_{pj} > 0\\
    C(p,m) &:= c(p)^\top m = 0\!\!\!\! \mod 2\pi \\
    A(i,j,p,m) &:= A(p)m < 0\\
	NND(i,j,p,m) &:=  Z(i,j) \vee (S(i,j,p) \wedge C(p,m) \wedge A(i,j,p,m) ).
\end{align*}

Note that $A(i,j,p,m)$ depends on $i$ and $j$ because we only care about the eigenvalue blocks where the coefficient is non-zero.

Writing out the predicate ENN in full with new notation, we have
\begin{align*}
    ENN &:= \exists m \; \forall i \forall j (u^{ij}_n(m) \text{ is identically zero} \vee \exists p \; PD_p(u^{ij}_n(m)) )\\
    &= \exists m \; \forall i \forall j ( Z(i,j) \vee \exists p \; (S(i,j,p) \wedge C(p,m) \wedge A(i,j,p,m) ) )\\
	&= \exists m \; \forall i \forall j  \; \exists p \; NND(i,j,p,m) 
\end{align*}
Now observe that the quantifications over $i$, $j$ and $p$ are finite and range over $1,\dots,d$. Thus we can replace the quantifications with finite disjunctions and conjunctions.

Our goal is to move all disjunctions outside the existential quantifier on $m$ so that we can use the IP-log algorithm from the previous section on conjunctions of terms of the form $Am < 0$.

Let $f$ range over functions assigning a particular choice of $p$ to each sequence $u_n^{ij}(m)$. There are $d^{(d \times d)}$ such functions.
\begin{align*}
    ENN &= \exists m \; \forall i \forall j ( \exists p \; NND(i,j,p,m) )\\
    &= \exists m \; \wedge_{ij} ( \vee_p \; NND(i,j,p,m) )\\
	&= \exists m \; \vee_f ( \wedge_{ij} NND(i,j,p = f(i,j),m) )\\
	&= \vee_f ( \; \exists m \; \wedge_{ij}  NND(i,j,p = f(i,j),m) )
.\end{align*}

Essentially this means we `non-deterministically' choose a $p$ for each sequence $u_n^{ij}(m)$ and then check if there is an $m$ that works for all of them --- that makes all the selected $p$'s into real positively dominating terms.

Analysing the elements of $NND(i,j,f(i,j),m)$, we see that $Z(i,j)$ and $S(i,j,f(i,j))$ constraints are trivially checkable. Constraints $\wedge_{ij} C(f(i,j),m)$ can be removed iteratively using Masser's theorem. The remaining conjunctions are terms of the form $\wedge_{ij} A(i,j,f(i,j),m)$, but since these are linear programs $\boldsymbol{A}(i,j,p)\boldsymbol{m} < \boldsymbol{0}$ we can simply concatenate the various matrices $\boldsymbol{A}(i,j,p)$ together. We can then use the IP-log algorithm from the previous section to check if there exists an integer solution to the conjunction of these terms.

Of course in practice we would only need to check $d$ different matrices of size $d-1 \times k$ since the eigenvalues are the same for all sequences.

Thus we have reduced the diagonalizable case to a disjunction of predicates, each of which can be reduced to strict homogenous IP-log. More generally, we have:

 \begin{theorem}
 The non-negative membership problem is decidable for commutative semigroups, assuming Schanuel's conjecture is true.
 \end{theorem}
 \begin{proof}
     See Appendix \ref{sec:nondiag}
 \end{proof}

\section{Undecidability in the Non-commutative Case}
To complement the decidability results for commuting matrices in the preceding sections, in this section we show undecidability of the full version of the membership problem, in which commutativity is not assumed:
\begin{problem}[Non-negative Membership]
	Given a set of $d \times d$ matrices with rational entries, decide whether the generated semigroup contains a non-negative matrix.
\end{problem}
The proof of undecidability is by reduction from the threshold problem for probabilistic automata, which is well-known to be undecidable \cite{DBLP:journals/siglog/Fijalkow17}.

\begin{problem}[Threshold Problem for Probabilistic Automata]
	Given vectors $u$ and $v$ in $\Q^d$ and a matrix semigroup $\mathcal{S}$ generated by stochastic matrices $\set{A_1, \dots, A_k} \in \Q^{d \times d}$, decide whether there exists a matrix $A \in \mathcal{S}$ such that $u^\top A v \geq 1/2$.
\end{problem}

\begin{theorem}
	\label{nms-undecidable}
	The Non-negative Membership Problem is undecidable.
\end{theorem}
\begin{proof}
Given non-negative integers $m,n$, write $0^{m\times n}$ for the zero matrix of dimension $m\times n$.  Suppose that we are
given an instance of the threshold problem for probabilistic automata, defined by vectors $u,v\in \Q^d$ and a matrix semigroup $\mathcal{S}$ generated by stochastic matrices $A_1, \dots, A_k \in \Q^{d \times d}$.  Now
consider the semigroup $\mathcal{S}'$ generated by 
the following matrices of dimension $(d+2)\times (d+2)$:
\begin{align*}
	U &:= \begin{bmatrix}
		1 & -1/2 & \boldsymbol{u}^\top \\
		0 & 0 & 0^{1\times d}  \\
		0^{d\times 1} & 0^{d\times 1} & 0^{d \times d}\\
	\end{bmatrix} \\
	A'_i &:= \begin{bmatrix}
		1 & -1/2 & 0^{1\times d}\\
		0 & 0 & 0^{1\times d}\\
		0^{d\times 1} & 0^{d\times 1} & {A}_i\\
	\end{bmatrix} \quad (i=1,\ldots,k)\\
	V &:= \begin{bmatrix}
		1 & -1/2 & 0^{1\times d} \\
		0 & 0 & 0^{1\times d}  \\
		0^{d\times 1} & \boldsymbol{v} & 0^{d \times d}\\
	\end{bmatrix}
\end{align*}
Note that matrix $A'_i$ incorporates $A_i$ for $i=1,\ldots,k$, while the matrices $U$ and $V$ respectively incorporate the initial and final vectors $u$ and $v$ of the probabilistic automaton.

We claim that the semigroup $\mathcal{S}'$
contains a non-negative matrix if and only if there exists a matrix $A \in \mathcal{S}$ such that $u^\top A v \geq \frac{1}{2}$.  To this end,
consider a string of matrices chosen from the set 
$\{ U, V\} \cup \{A'_1,\ldots, A'_k\}$.
Any product $B$ of such a string that does not end with a suffix 
$U A'_{i_1} \cdots A'_{i_s} V$ for $i_1,\ldots,i_s \in \{1,\ldots,k\}$, has $B_{1,2}=-\frac{1}{2}$ and hence cannot be a non-negative matrix.  
It remains to consider products $B$ of strings that do have such a suffix.
In this case we have 
\[ B_{1,2} = (UA'_{i_1} \cdots A'_{i_s}V)_{1,2} = u^\top A_{i_1} \cdots A_{i_s} v  - \frac{1}{2}\]
and hence $B$ is only non-negative if 
$u^\top A_{i_1} \cdots A_{i_s} v \geq \frac{1}{2}$.   Since it further holds that 
$B_{1,1}=1$ and $B_{i,j}=0$ for all other entries $(i,j)$, we conclude that exists a non-negative matrix in the semigroup $\mathcal{S}'$ if and only if there exists a matrix $A \in \mathcal{S}$ such that $u^\top A v \geq \frac{1}{2}$.
\end{proof}
\section{Further work}
We leave open the question of quantitative refinements of our decidability results. These include giving complexity upper bounds for the Non-negative and Positive Membership Problems as well as the related question of giving upper bounds on the length of the shortest string of generators that yields a non-negative or positive matrix in a given semigroup.  Both questions would seem to be difficult owing to the use of Schanuel's Conjecture in our proofs.   Characterising the complexity of determining 
whether a matrix is eventually non-negative would seem to be more straightforward.  We claim
that the decision procedure can be implemented in non-deterministic polynomial time.
Note that the analogous Eventual Positivity problem is in $\PTIME$  \cite{noutsos2006perron}.

\bibliography{references}

\newpage
\appendix

\section{Non-diagonalizable Case}
\label{sec:nondiag}

\begin{problem}[Non-negative Membership for Commutative Semigroup]
	Given a set of commuting $d \times d$ matrices $\set{A_1, \dots, A_k}$ with rational entries, decide whether the semigroup generated by multiplying these matrices together contains a matrix with all its entries greater than or equal to zero.
\end{problem}

It is known that unfortunately simultaneous Jordanization of commuting matrices is not always possible \cite{Cai94}. However, a slightly weaker block diagonal form \cite[Thm 12]{joao2016solvability} is possible. Here we put the matrices into block diagonal form, where each block is of the form $\lambda_iI + N$ where $N$ is strictly upper triangular and thus nilpotent.

Let $A_1,\dots,A_k \in \Q^{d \times d}$ be a set of commuting matrices. 

Define the $(m_1,\dots,m_k)$-integer parameterized matrix entry recurrence $u_n^{ij}(m_1,\dots,m_k) := e_i^\top [A_1^{m_1} \dots A_k^{m_k}]^n e_j$.

Let $S$ be a matrix that simultaneously block-diagonalizes the matrices such that $A_r = S^{-1} BD_r S = S^{-1} \diag(B_{r1},\dots,B_{rb}) S$. Here $B_{rl} = \lambda_{rl}I + N_{rl}$ denotes the $l$th block of $A_r$, where $N_{rl}$ is strictly upper triangular. Note that for fixed $l$ the various $N_{rl}$ inherit commutativity from the original matrices.
Then we have that 
\begin{align*}
	u_n^{ij}(m_1,\dots,m_k) &:= e_i^\top [A_1^{m_1} \dots A_k^{m_k}]^n e_j \\
	&= e_i^\top S^{-1} [BD_1^{m_1} \dots BD_k^{m_k}]^n S e_j \\
	&= e_i^\top S^{-1} [\diag(B_{11},\dots,B_{1b})^{m_1} \dots \diag(B_{k1},\dots,B_{kb})^{m_k}]^n S e_j \\
	&= e_i^\top S^{-1} \left[\diag \left( \prod_{r=1}^k B_{r1}^{m_r}, \dots, \prod_{r=1}^k B_{rb}^{m_r} \right )\right]^n S e_j \\
	&= [s^{-1}_{i1},\dots,s^{-1}_{id}] \diag \left( \left[\prod_{r=1}^k B_{r1}^{m_r}\right]^n, \dots, \left[\prod_{r=1}^k B_{rb}^{m_r}\right]^n \right )^n [s_{1j},\dots,s_{dj}]^\top .
\end{align*}

Let us examine the structure of the submatrix $\left[\prod_{r=1}^k B_{rb}^{m_r}\right]^n$. 
Recall that $B_{rl} = \lambda_{rl}I + N_{rl}$. Note that any power or product of powers of the nilpotents can be non-zero only upto total degree at most $d$. We adopt the convention that a zero power indicates the identity matrix of appropriate size.
For ease of notation we drop the block subscript and expand out
\begin{align*}
	\left[\prod_{r=1}^k B_{rb}^{m_r}\right]^n &= \left[\prod_{r=1}^k (\lambda_rI + N_r)^{m_r}\right]^n \\
	&= \prod_{r=1}^k \lambda_r^{nm_r}(I + N_r/\lambda_r)^{nm_r} \\
 &= \prod_{r=1}^k \lambda_r^{nm_r} \cdot \prod_{r=1}^k \left( \sum_{i_r = 0}^d \binom{nm_r}{i_r} (N_r/\lambda_r)^{i_r} \right) \\
 &= \prod_{r=1}^k \lambda_r^{nm_r} \cdot \left[ \sum_{(i_1,\dots,i_k) = (0,\dots,0)}^{i_1+\dots+i_k = d} \left( \prod_{r =1}^k\binom{nm_r}{i_r} (N_r/\lambda_r)^{i_r}  \right)\right] \\
 &= \MPoly(n\mathbf{m}) \prod_{r=1}^k \lambda_r^{nm_r}  .
\end{align*}
Here $\MPoly(n\mathbf{m})$ is shorthand for the block-matrix with entries which are polynomial of degree at most $d$ in the variables $(m_{1},\dots,m_{k})$ (scaled by $n$) with coefficients arising from the nilpotent submatrices. 

Substituting this back into our original expression for $u_n^{ij}(m_1,\dots,m_k)$ we get

\begin{align*}
	&u_n^{ij}(m_1,\dots,m_k) := e_i^\top [A_1^{m_1} \dots A_k^{m_k}]^n e_j \\
	&= [s^{-1}_{i1},\dots,s^{-1}_{id}] \diag \left( \left[\prod_{r=1}^k B_{r1}^{m_r}\right]^n, \dots, \left[\prod_{r=1}^k B_{rb}^{m_r}\right]^n \right )^n [s_{1j},\dots,s_{dj}]^\top \\
	&= [s^{-1}_{i1},\dots,s^{-1}_{id}] \diag \left( \MPoly_1(n\mathbf{m}) \prod_{r=1}^k \lambda_{r1}^{nm_r}, \dots, \MPoly_b(n\mathbf{m_r}) \prod_{r=1}^k \lambda_{rd}^{nm_r} \right )^n [s_{1j},\dots,s_{dj}]^\top \\
	&= \sum_{l=1}^d  \left( \poly^{ij}_l(n\mathbf{m}) \prod_{r=1}^k\lambda_{rl}^{nm_r} \right) \text{(after folding in constants)}.
\end{align*}


Notice that the asymptotic top term in $n$ in the polynomial $\poly^{ij}_l(n\mathbf{m})$ is the homogenous subpolynomial of highest degree in $(m_{1},\dots,m_{k})$ - call it $h^{ij}_l(\mathbf{m})$.
Once we pick a particular $p$ to be our positively dominating term, we only need the following three conditions for the recurrence to be positively dominated by $p$:

\begin{enumerate}
	\item $h^{ij}_p(\mathbf{m})> 0,$\\
	\item $\prod_{r=1}^k\lambda_{rp}^{m_r} > 0 ,$\\
	\item $\forall l \neq p \text{ such that } \poly^{ij}_l(n\mathbf{m}) \text{ is not the zero polynomial,}  \; \abs{\prod_{r=1}^k\lambda_{rp}^{m_r}} > \abs{\prod_{r=1}^k\lambda_{rl}^{m_r}} $.
\end{enumerate}

Via the same algebraic manipulations as in the diagonalizable case, checking $$ENN := \exists m \; \forall i \forall j (u^{ij}_n(m) \text{ is identically zero} \vee \exists p \; PD_p(u^{ij}_n(m)) )$$ reduces to solving conjunctions of the form $$\wedge_{ij} (h^{ij}_l(\mathbf{m})> 0 \wedge \mathbf{c}(p)^\top \mathbf{m} = 0\!\!\!\! \mod 2\pi \wedge \mathbf{A}(p)\mathbf{m} < 0)$$ over the integers.

We can eliminate the second conjunct using Masser's theorem. Now suppose there exists a real solution on the unit sphere to $h^{ij}_l(\mathbf{m})> 0 \wedge \mathbf{A}(p)\mathbf{m} < 0$. By openness, there exists a rational solution nearby. 
Since both these conjuncts are homogenous, $\mathbf{m}$ is a solution iff $n\mathbf{m}$ is a solution, for all real $n > 0$. Thus we may clear denominators from the rational solution to obtain an integer solution $\mathbf{m}$ to $ENN$. 

The sentence $$\exists \mathbf{m} \in \R^k : h^{ij}_l(\mathbf{m})> 0 \wedge \mathbf{A}(p)\mathbf{m} < 0 $$ can be written in the first order theory of the reals with exponentiation, which is decidable assuming Schanuel's conjecture as shown by Wilkie and Macintyre \cite{WM}.

We now need to prove that failing to find such a real solution implies that the semigroup does not contain a non-negative matrix.
Suppose that there exists some $m = (m_1,\dots,m_k)$ such that $A_1^{m_1}\dots A_k^{m_k}$ is non-negative. Then $(A_1^{m_1} \dots A_k^{m_k})^n$ is non-negative for all $n$. 
By Proposition \ref{prop:DOM3}, each individual recurrence is ultimately zero or must have a strictly dominant top term. Thus $m$ satisfies $Am < 0$ and $c^\top m = 0\!\! \mod 2\pi$ for the appropriate $A$ and $c$.
The top term (as a function of $n$) has a polynomial coefficient which is the homogenous polynomial we identified above. Thus the homogenous polynomial is non-negative for $m$, completing the requirements necessary for the sentence $\exists \mathbf{m} \in \R^k : h^{ij}_l(\mathbf{m})> 0 \wedge \mathbf{A}(p)\mathbf{m} < 0 $ to be true. 

We conclude:

\begin{theorem}
 The non-negative membership problem is decidable for commutative semigroups, assuming Schanuel's conjecture is true.
 \end{theorem}

 \section{Positive Rational Sequences are Dominated}
 \label{sec:positive-dom}

 We need the following classical results from Perron-Frobenius theory (see, e.g.,~\cite[Chap. 8]{Meyer2023}).
\begin{theorem}
	\label{pf}
	\begin{enumerate}
		\item If $A \geq 0$ is irreducible then it has \emph{cyclic peripheral spectrum}, i.e., its eigenvalues of maximum modulus have the form $\{ \rho ,\rho \omega, \ldots ,\rho\omega^{k-1}\}$, where $\rho>0$, $k$ is a positive integer, and $\omega$ is a primitive $k$-th root of unity. If $A$ has only one eigenvalue on the spectral circle it is called a primitive matrix.
		\item For a non-negative irreducible primitive matrix $A$ the pointwise limit $\lim_{n \rightarrow \infty} (A/\rho(A))^n$
		exists and is
  a strictly positive matrix.
	\end{enumerate}
\end{theorem}

We now prove Berstel's result for matrix entry recurrences.

\begin{proposition}
	\label{pdr}
 Let $M\in \Q^{d\times d}$ be a non-negative non-degenerate matrix.
 Then for all $i,j \in \{1,\ldots,d\}$ the LRS $u_n = (M^n)_{i,j}$ is dominated.
\end{proposition}

\begin{proof}
	Since $M \geq 0$ there exists a permutation matrix $P$ such that
 $M$ can be written in the form
 $M = PUP^{-1}$, where $U \geq 0$ is block upper triangular.   
 It follows that there exist $i',j' \in \{1,\ldots,d\}$ such that 
 \[ (M^n)_{i,j} = (P U^n P^{-1})_{i,j} =
    (U^n)_{i',j'} \]
    for all $n\in \mathbb{N}$.  Now write
	$$U =
	\begin{pmatrix}
		B_{1,1} & B_{1,2} & \ldots & B_{1,e} \\
		0 & B_{2,2} & \ldots  & B_{2,e} \\
		0 & 0 & \ddots & \vdots  \\
		0 & 0 & 0 & {B_{e,e}}
	\end{pmatrix}\, , $$
where all the blocks in $U$ are non-negative and the diagonal blocks 
$B_{1,1},B_{2,2},\ldots,B_{e,e}$ are irreducible.  Then
 \begin{gather}
 (U^n)_{i',j'} = \sum_{\substack{l_1<l_2<\cdots <l_m\\n_1+n_2+\cdots n_m = n - (m-1)}} e_{i'}^\top \, B_{l_1,l_1}^{n_1}B_{l_1,l_2} B_{l_2,l_2}^{n_2} \cdots B_{l_{m-1},l_m} B_{l_m,l_m}^{n_m}  \, e_{j'} \, ,
 \label{eq:long}
 \end{gather} 
 where the sum runs over all positive integers $m$ and strictly increasing sequences of block indices $l_1 < \cdots < l_m$.
	
	Consider a single block $B_{l,l}$ along the diagonal. Since it is irreducible and non-negative, it has cyclic peripheral spectrum. By our assumption of non-degeneracy, $r_l := \rho(B_{l,}) \geq 0$ is the only eigenvalue on the spectral circle. Thus $B_{l,l}$ is primitive and by our second Perron-Frobenius result above, asymptotically $B_{l,l}^n/r_l^n \sim  C_l$ where $C_l$ is a positive matrix and the asymptotic equivalence relation $\sim$ applies entry-wise. Let $r_{max}$ be the maximum spectral radius of a block $B_{l,l}$ lying on a path from $i'$ to $j'$.

We now analyse the asymptotic behavior of the normalized recurrence $(U^n)_{i',j'}/r_{max}^n$.
 Consider a summand $S_n$ in $(U^n)_{i',j'}/r_{max}^n$. Replacing the diagonal blocks with their asymptotic limits, $$S_n \sim  \left(\frac{r_{l_1}}{r_{max}}\right)^{n_1}\left(\frac{r_{l_2}}{r_{max}}\right)^{n_2}\dots \left(\frac{r_{l_m}}{r_{max}}\right)^{n_m} e_{i'}^\top \,C_{l_1}B_{l_1,l_2} C_{l_2} \cdots B_{l_{m-1},l_m} C_{l_m}  \, e_{j'}$$
 Although the number of terms in the sum grows polynomially in $n$, we see that each term with some $r_{l_k} < r_{max}$ that does not have $n_k$ constant tends to zero exponentially quickly.
	The remaining summands in $(U^n)_{i',j'}/r_{max}^n$
 are thus those where $n_k$ is non-constant only for blocks with $\rho(B_{l_k,l_k}) = r_{max}$.
Let $K$ be the sum of the constant powers for non-maximal blocks in such a summand $Q_n$, and $P$ be the product of the powers of non-maximal $r_{l_k}$. Then $$Q_n \sim r_{max}^{n} \cdot (P/r_{max}^{K+m-1}) \cdot e_{i'}^\top \,C_{l_1}B_{l_1,l_2} C_{l_2} \cdots B_{l_{m-1},l_m} C_{l_m}  \, e_{j'}$$.
	Observe that the coefficient of $r_{max}^n$ is a product of spectral radii and non-negative matrices, and is thus non-negative. This implies different such terms containing $r_{max}^n$ cannot cancel out. So $e_{i'}^\top U^n e_{j'} \sim A(n) \cdot r_{max}^n$ for some polynomial $A$
 with positive coefficients depending on the non-diagonal blocks and the spectral radii of the diagonal blocks. Thus $(M^n)_{i,j}$ is either dominated by $r_{max}$ or ultimately zero (the latter in case $r_{max} = 0$ or the sum in~\eqref{eq:long} is empty).
\end{proof}
	
\end{document}